\documentclass[]{article}
\usepackage{amsfonts}
\usepackage{amssymb}
\usepackage{graphics}
\usepackage{graphicx}
\usepackage{makeidx}
\usepackage{latexsym}
\usepackage[dvips]{color}
\usepackage[centertags]{amsmath}
\usepackage{amsthm}
\usepackage{newlfont}
\usepackage{stmaryrd}
\usepackage{mathrsfs}
\usepackage{enumerate}
\usepackage[english]{babel}

\newcommand{\be}{\begin{eqnarray}}
\newcommand{\ee}{\end{eqnarray}}
\usepackage{tikz}
\usetikzlibrary{arrows}
\newtheorem{theorem}{Theorem}[section]

\newtheorem{proposition}[theorem]{Proposition}

\usepackage{hyperref}
\usepackage{caption}
\usepackage{subcaption}

\begin{document}

\title{{\bf Low temperature behavior of\\ nonequilibrium multilevel systems\\ }}

\author{Christian Maes$^1$, Karel~Neto\v{c}n\'{y}$^2$ and  Winny O'Kelly de Galway$^1$\\
$^1$ Instituut voor Theoretische Fysica, KU Leuven \\
$^2$ Institute of Physics AS CR, Prague, Czech Republic}

\maketitle
\begin{abstract}
We give a low temperature formula for the stationary occupations in Markovian systems away from detailed balance. Two applications are discussed, one to determine the direction of the ratchet current and one on population inversion.  Both can take advantage of low temperature to improve the gain and typical nonequilibrium features. 
The new formula brings to the foreground the importance of kinetic aspects in terms of reactivities for deciding the levels with highest occupation and thus gives a detailed quantitative meaning to Landauer's blowtorch theorem at low temperature. 
\end{abstract}
\baselineskip=20pt
\section{Introduction}
Multilevel systems and their dynamical characterization arise from different sources.  On the one hand, they describe the weak coupling limit of small quantum systems in contact with an environment consisting of large equilibrium baths or radiation.  When projecting the resulting open quantum dynamics on the energy basis, we get a jump process with rates that are determined by Fermi's Golden Rule and the evolution follows the Pauli Master equation.  The Markov limit (after the Born-Oppenheimer approximation) excludes mostly the effects of quantum coherence, but relaxation and nonequilibrium steady regimes can still be investigated \cite{glau,ls}.  On the other hand, Markov jump processes also arise as effective dynamics on local minima of a thermodynamic potential from diffusion processes along the (Arrhenius-Eyring-)Kramers reaction rate theory \cite{pol,han,berg}.   The diffusion process itself derives from classical mechanics after suitable rescaling in a coarse-grained description.\\
These different derivations of Markov jump processes remain sometimes valid in the case of driven dynamics, in the presence of nonequilibrium forces or when coupled to different reservoirs.  In fact also conceptually they often continue to play a major role in explorations of nonequilibrium physics, and they are used as models of nonequilibrium systems, even to illustrate fundamental points in a nonequilibrium theory e.g. of fluctuation and response.  By the exponential form of the rates, they are strongly nonlinear in potential differences and external forces and that is generally felt as an important point in the study of far-from-equilibrium systems.  More than diffusion processes they also play in modeling of chemical reactions and aspects of life processes.\\

Very efficient computational tools exist for the solution of the Master equation and over the last decades a major part of their study has been numerical or via simulation methods. 
There is however no systematic physical understanding of the solution of the Master equation away from detailed balance. In the present paper we concentrate on the analytical theory of the low temperature asymptotics in that steady nonequilibrium regime.  We assume indeed that the physical system under consideration is in contact with an environment at uniform temperature while also subject to external forcing that will break the Boltzmann occupation statistics.  Not only does one expect that nonequilibrium features can become more prominent at low temperatures, but also, somewhat in analogy with equilibrium statistical mechanics, that a more explicit asymptotic analysis is possible there as perturbation of the zero temperature phase diagram.\\
The present paper is a continuation of that program started in \cite{heatb,lowt}.  The main formula in Section \ref{Mot} gives a low temperature expression for the occupation of a general multilevel system.  The difference with \cite{heatb} is that we include now the role of the reactivity or Arrhenius prefactors in the pre-asymptotic form.  We give two applications where these reactivities matter which are discussed in Sections \ref{rat} and \ref{laser}.  One is (a continuous time version of)
Parrondo's game at low temperature, \cite{par}.  That is a (flashing potential) ratchet which is studied here for determining the direction of the ratchet current, a general problem which cannot be decided by entropic considerations only.  A second application concerns the role of equalizers in laser working.  These equalizers contribute directly to the reactivities and are responsible for the necessary population inversion in laser working, In particular we search for the influence of the equalizer in deciding the most occupied energy level, and thus influencing the gain.\\

\section{Low temperature formula}\label{Mot}
Low temperature analysis of nonequilibria requires finding the analogue to the so called ground states for statistical mechanical systems at equilibrium. The analysis in \cite{heatb} did not distinguish between such dominant states and a finite (low) temperature formula for the level occupations is needed.  In this section we propose the required extension of the Kirchhoff-Freidlin-Wentzel formula considered in \cite{heatb}.  The correction takes into account some of the kinetics of the model and will therefore also enable the study of (finite temperature) ``noise'' effects that contribute to the nonequilibrium phenomenology.\\

The multilevel system is determined by a finite number of states $x,y,\ldots \in K$, sometimes called energy levels as they often correspond to minima of some potential function or to the energy spectrum $E(x), E(y),\ldots$ of a finite quantum system. (Ignoring possible degeneracies, we do not distinguish here between states and levels.)\\ 
We consider an irreducible continuous time Markov jump process on $K$ with transition rates $\lambda(x,y) \equiv\lambda(x,y;\beta)$ for the jump $x\rightarrow y$ that depend on a parameter $\beta\geq 0$, interpreted as the inverse temperature. In the absence of other reservoirs or of nonequilibrium forces the dynamics satisfies the condition of detailed balance $\lambda(x,y)/\lambda(y,x) = \exp -\beta [E(y)-E(x)]$ for allowed transitions, and the stationary occupation is then $\rho_{\mbox{eq}}(x) \propto \exp-\beta E(x)$ giving the Boltzmann equilibrium statistics.  Nonequilibrium conditions still preserve the dynamical reversibility $\lambda(x,y)\neq 0 \implies \lambda(y,x)\neq 0$ but  detailed balance is broken. We are looking then for a low temperature characterization of the nonequilibrium occupation.\\
After   \cite{heatb} we assume the existence of the limit
\be\label{eq:asymp}
\lim_{\beta\rightarrow \infty} \frac{1}{\beta} \log\lambda(x,y) =:  \phi(x,y)
\ee
which is abbreviated as $\lambda(x,y)\asymp e^{\beta \phi(x,y)}$.  Similarly we define the log-asymptotic life time $\tau(x)\asymp e^{\beta \Gamma(x)}$ from the escape rates $\sum_y \lambda(x,y)$,
\[
\Gamma(x):= -\lim_{\beta \rightarrow \infty} \frac{1}{\beta}\log \sum_y \lambda(x,y)= -\max_y \phi(x,y) 
\]
The log-asymptotic transition probability is then $e^{-\beta U(x,y)}$ with 
\[
U(x,y):= -\Gamma(x) -\phi(x,y)
\]
The larger $U(x,y)$ the more improbable the transition from $x$ to $y$ becomes. Clearly, $U(x,y)\geq 0$ and for all $x$ there is at least one state $y\neq x$ for which $U(x,y)=0$. All $y$ for which $U(x,y)=0$ are called preferred successors of $x$. That defines the digraph $K^D$ with $K$ as vertex set and  having directed arcs (oriented edges) indicating all preferred successors.\\ 
We now define the reactivities $a(x,y)$, which will often be effective prefactors in a formula for reaction rates:
\be\label{rates}
\lambda(x,y)= a(x,y) e^{-\beta[\Gamma(x)+U(x,y)]},\quad a(x,y):=\lambda(x,y)\,e^{-\beta \phi(x,y)}=e^{o(\beta)}
\ee
as the sub-exponential part of $\lambda(x,y)$.\\
The new low temperature formula is for the stationary distribution $\rho(x), x\in K$, solution of the stationary Master equation
$\sum_y[\lambda(x,y)\rho(x) - \lambda(y,x)\rho(y)]=0, y\in K$.   We derive it from the Kirchhoff formula on the graph $\mathcal{G}$ which has $K$ as vertex set and with edges between any $x,y\in K$ where $\lambda(x,y)\neq 0 $ ($\iff \lambda(y,x)\neq 0$) independent of $\beta$.  In what follows, spanning trees ${\mathcal T}$ refer to that graph $\mathcal{G}$, and $\mathcal{T}_x$ denotes the in-tree to $x$ defined for any tree $\mathcal{T}$  by orienting every edge in $\mathcal{T}$ towards $x$.

\begin{theorem}\label{thm1}
There is $\varepsilon >0$ so that as
$\beta\rightarrow \infty$,
\be\label{eq:as2}
\rho(x)=\frac{1}{\mathcal{Z}} A(x)e^{\beta[\Gamma(x)-\Theta(x)]}\,(1+O(e^{-\beta\varepsilon}))
\ee
with 
\begin{eqnarray}
\Theta(x) &:=& \min_{\mathcal{T}}U(\mathcal{T}_x) \quad\mbox{ for } \quad
U(\mathcal{T}_x) := \sum_{(y,y')\in \mathcal{T}_x} U(y,y') \quad\mbox{ and }\label{thth}\\
A(x)& :=& \sum_{\mathcal{T}\in M(x)}\prod_{(y,y')\in \mathcal{T}_x} a(y,y')=e^{o(\beta)}\label{aa}
\end{eqnarray} where the last sum runs over all spanning trees minimizing $U(\mathcal{T}_x)$ (i.e., $\mathcal{T}\in M(x)$ if $\Theta(x) = U(\mathcal{T}_x)$). 
\end{theorem} 
Before we give the proof (at the end of this Section) there are a number of remarks and illustrations.\\
The various terms in formula \eqref{eq:as2} have a physical interpretation in terms of life time $\Gamma$ and accessibility $\Theta$ of states, cf. \cite{lowt} for an analysis of the boundary driven Kawasaki process.
To be more specific rewrite the rates in the Arrhenius form $\lambda(x,y) = a(x,y) \exp \beta[-\Delta(x,y)+E(x)]$ with $a(x,y)=a(y,x)$ the symmetric prefactor which has no exponential dependence on temperature.  The energy $E(x)$ and barrier height $\Delta(x,y)$ refer to a thermodynamic potential landscape.  Of course that picture fails to a great extent when the system undergoes nonequilibrium driving and we should then think of the antisymmetric part in the barrier height $\Delta(x,y)- \Delta(y,x)$ as the work of  non-conservative forces; see Fig.~\ref{landscape}. 

\begin{figure}[h]
\centering
		\includegraphics[scale=1]{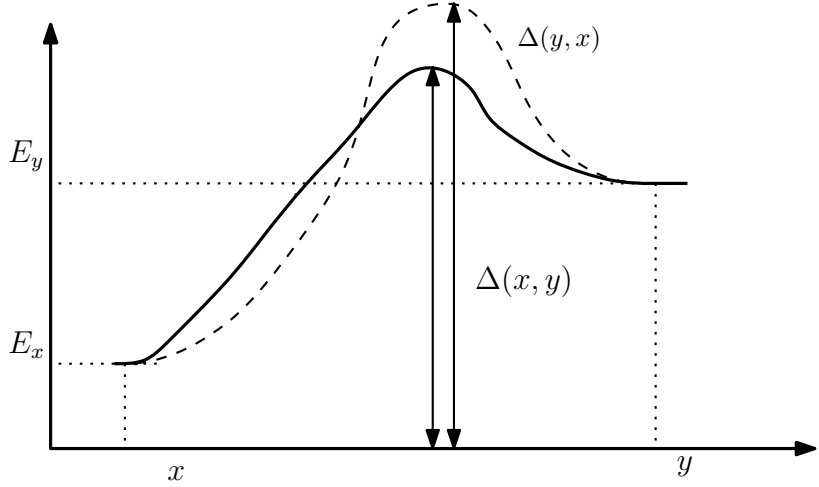}
\caption{Climbing the barrier from $x$ need not be the same as climbing it from $y$, as encoded in the work $\Delta(x,y)-\Delta(y,x)$.}
\label{landscape}
\end{figure}

 The local detailed balance form of the rates is
\begin{eqnarray}\label{ldb}
\lambda(x,y) &=& \psi(x,y)\,e^{s(x,y)/2}\\
a(x,y): = \psi(x,y)\,e^{-\frac{\beta}{2}[E(x)+E(y) - \Delta(x,y)-\Delta(y,x)]},&& s(x,y): = \beta[E(x)-E(y) + \Delta(y,x)-\Delta(x,y)]\nonumber
\end{eqnarray}
where $s(x,y) = -s(y,x)$ is the entropy flux to the environment at inverse temperature $\beta$ during $x\rightarrow y$. We see that the $a(x,y)$ truly pick up kinetic (non-thermodynamic) aspects of the transition rates.
For \eqref{eq:asymp} we have $\phi(x,y) = E(x) - \Delta(x,y)$.
  In the case
of detailed balance, $\Delta(x,y)=\Delta(y,x)$ which implies that $\Theta(x)$ is a constant and $M(x)$ also does not depend on $x$ making $A(x)$ constant.  Nonequilibrium is responsible for having  a state dependent set of minimizing trees. In other words it is exactly for nonequilibrium that the reactivities start to matter in the occupation statistics.  That is a version of Landauer's blowtorch theorem, \cite{lan}, \cite{heatb}.  We illustrate it with a simple model.

\textit{Example}. Let us consider a nearest neighbor random walk on a ring $K = \{1,2,\hdots, N\}, N>2,$ where $N+1\equiv 1$. The transition rates are  
\be\label{ext}
\lambda(x,x+1)=a_xe^{\frac{\beta}{2}q},\quad \lambda(x+1,x)=a_xe^{-\frac{\beta}{2}q},\quad q>0
\ee
where $0<a_x = e^{o(\beta)}$. The digraph  is strongly connected; see Fig.~\ref{digraph}.

\begin{figure}[h] 
\centerline{
\begin{tikzpicture}
\node [draw, shape= circle,label=320:$1$] (E1) at (270:2) {$$};
\node  [draw, shape= circle,label=320:$2$ ](E2) at (315 :2) {$$};
\node  [draw, shape= circle,label=0:$3$ ](E3) at (0:2) {$$};
\node [draw, shape= circle,label=225:$N$ ]  (EN) at (225:2) {$$};
\node [draw, shape= circle,label=180:$N-1$  ] (EN1) at (180:2) {$$};
\node [draw, shape= circle,label=135:$N-2$  ] (EN2) at (135:2) {$$};
\node  [draw, shape= circle ](DOT1)  at (90:2) {};
\node  [draw, shape= circle ](DOT2)  at (45:2) {};
\foreach \from/\to in {E1/E2, E2/E3, EN/E1, EN2/EN1, EN1/EN}
\draw [-> , thick] (\from) -- (\to);
\draw[->, dotted,thick] (DOT2) -- (DOT1);
\draw[->, dotted,thick] (E3) -- (DOT2);
\draw[->, dotted,thick] (DOT1) -- (EN2);
\end{tikzpicture}
}
\caption{Digraph for the dynamics with rates \eqref{ext}.  For any pair $\{x,y\}$ there is a connection $x\longrightarrow y$.}
\label{digraph}
\end{figure}
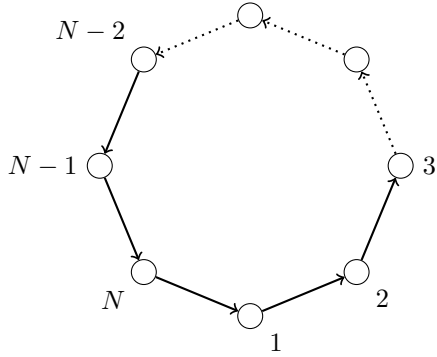

From \eqref{thth} that implies that $\Theta(x)=0$ for all states.  Furthermore $\Gamma(x) = -q $, so that \eqref{eq:as2} gives
\be
\frac{\rho(x)}{\rho(y)}=\frac{a_y}{a_x}+O(e^{-\beta \epsilon})
\ee
and the most occupied states are naturally those that have lowest reactivities in the for thermodynamic reasons fastest direction (here: $x\rightarrow x+1$) at low temperatures.  In the following two sections follow more  examples  where the subasymptotic part $A(x)$ in \eqref{eq:as2} \textit{really} matters for nonequilibrium analysis.\\

A final remark is that \eqref{aa} can be written as the determinant of an $x-$dependent matrix determined by putting weights $a(y,y')$ on the edges $(y,y')$ of the digraph $K^D$, which is a consequence of the matrix-tree theorem \cite{Tutte}. The example in Section \ref{rat} provides an illustration.

%$A(x) = \det L_x$ 

\begin{proof}[Proof of Theorem \ref{thm1}]
To find the stationary distribution we apply the Kirchhoff formula on the graph $\mathcal{G}$:
\be
\rho(x)=\frac{W(x)}{\sum_yW(y)}, \quad  W(x):=\sum_\mathcal{T}w(\mathcal{T}_x)
\ee
in which the weight $w(\mathcal{T}_x)$ is
 the product of transition rates $\lambda(y,z)$ over all oriented edges $(y,z)$ in the in-tree $\mathcal{T}_x$:
  \be
w(\mathcal{T}_x):=\prod_{(y,z)\in \mathcal{T}_x}\lambda(y,z)
\ee We substitute \eqref{rates}:
\be\nonumber
 W(x)&=& \sum_{\mathcal{T}}\prod_{(y,z)\in \mathcal{T}_x} e^{-\beta(\Gamma(y)+U(y,z))} a(y,z)\\ \nonumber
&=& \sum_{\mathcal{T}} \left(\prod_{y\neq x} e^{-\beta\Gamma(y)}\right)\prod_{(y,z)\in \mathcal{T}_x} e^{-\beta U(y,z)} a(y,z)\\ \nonumber
&=& e^{\beta(\Gamma(x)-\Theta(x)) -\beta\sum_y \Gamma(y)}\sum_{\mathcal{T}} e^{-\beta(U(\mathcal{T}_x)-\Theta(x))} \prod_{(y,z)\in \mathcal{T}_x}a(y,z),
\ee
where $\Theta(x)=\min_{\mathcal{T}_x} U(\mathcal{T}_x)$. We split up the sum over spanning trees $\cal T$ into two parts by collecting first the trees in $M(x)$, the set of trees minimizing $U(\mathcal{T}_x)$, and summing over the rest. We then continue
\be\nonumber
W(x)&=&e^{\beta(\Gamma(x)-\Theta(x)) -\beta\sum_y \Gamma(y)}  \left(\sum_{\substack{\mathcal{T},\\ \mathcal{T}\in M(x)}} \prod_{(y,z)\in \mathcal{T}_x}a(y,z)+ \sum_{\substack{\mathcal{T},\\ \mathcal{T}\notin M(x)}}e^{-\beta(U(\mathcal{T}_x)-\Theta(x))}\prod_{(y,z)\in \mathcal{T}_x}a(y,z)\right)\\
&=&e^{\beta(\Gamma(x)-\Theta(x)) -\beta\sum_y \Gamma(y)}\sum_{\mathcal{T}\in M(x)}\prod_{(y,z)\in \mathcal{T}_x}a(y,z)  \left(1 +O(e^{-\beta \epsilon})\right)
\ee
where $\varepsilon :=\min_x\min_{\mathcal{T}, \mathcal{T}\notin M(x)}\left[ U(\mathcal{T}_x)-\Theta(x)\right]>0$. The normalisation is 
\be\nonumber
\sum_y W(y)&=&e^{-\beta\sum_z \Gamma(z)}\sum_y e^{\beta(\Gamma(y)-\Theta(y))} A(y)\left(1 +O(e^{-\beta \varepsilon})\right)
\ee
which ends the proof.
\end{proof}

\section{Low temperature ratchet current}\label{rat}

We consider here a version of the well-known Parrondo game in continuous time and for random flipping between a flat potential (free random walker) and a nontrivial energy landscape, \cite{par,conpar}.  We show how the formula \eqref{eq:as2} produces an expression for the low temperature ratchet current.  In particular its direction is not determined by entropic considerations (only) but involves the reactivities.\\

%As a first remark (for general systems), to know the low temperature behavior of a current $j(x,y)$ between states $x,y$ it is in general not sufficient to know the asymptotics $\lim_{\beta \rightarrow \infty} 
%\frac 1{\beta} \log \rho(x)$ as $j(x,y) = \rho(x) k(x,y) - \rho(y)k(y,x)$ is a sum.  That is why the formula \ref{eq:as2} at low but finite temperature is useful (and necessary).\\

The ratchet consists of two rings with the same number $N>2$ of states.  States are therefore denoted by $x= (i,n)$ where $i \in\{ 1=N+1,2,\ldots,N\}$ and $n=0,1$.

\begin{figure}[h]
\centerline{
\begin{tikzpicture}
\node [draw, shape= circle,label=90:$1$ ](x1) at (270:2)  {$$};
\node [draw, shape= circle,label=135:$2$] (x2) at (315 :2) {$$};
\node  [draw, shape= circle,label=180:$3$ ](x3) at (0:2) {$$};
\node [draw, shape= circle ,label=45:$N$ ] (xN) at (225:2) {$$};
\node  [draw, shape= circle,label=0:$N-1$ ](xN1) at (180:2) {$$};
\node  [draw, shape= circle,label=315:$N-2$ ](xN2) at (135:2) {$$};
\node  [draw, shape= circle ](xDOT1)  at (90:2) {};
\node [draw, shape= circle ] (xDOT2)  at (45:2) {};
\node [draw, shape= circle,label=320:$E_1$] (E1) at (270:4) {$$};
\node  [draw, shape= circle,label=320:$E_2$ ](E2) at (315 :4) {$$};
\node  [draw, shape= circle,label=0:$E_3$ ](E3) at (0:4) {$$};
\node [draw, shape= circle,label=225:$E_N$ ]  (EN) at (225:4) {$$};
\node [draw, shape= circle,label=180:$E_{N-1}$  ] (EN1) at (180:4) {$$};
\node [draw, shape= circle,label=135:$E_{N-2}$  ] (EN2) at (135:4) {$$};
\node  [draw, shape= circle ](DOT1)  at (90:4) {};
\node  [draw, shape= circle ](DOT2)  at (45:4) {};
\foreach \from/\to in {E1/E2, E2/E3, E1/EN, EN2/EN1}
\draw [<- , thick] (\from) -- (\to);
\draw[dotted,-> , thick] (EN2)-- (DOT1);
\draw[dotted,-> , thick] (DOT2)--(E3) ;
\draw[dotted,->, thick ] (DOT1)--(DOT2)  ;
\draw[dotted,->, thick ] (xDOT1)--(DOT1)  ;
\draw[dotted,->, thick ] (xDOT2)--(DOT2)  ;
\draw[dotted, thick, <-> ] (xDOT1) -- (xN2);
\draw[dotted, thick, <-> ] (xDOT1) -- (xDOT2);
\draw[dotted, thick, <-> ] (xDOT2) -- (x3);
\draw[ thick, <-> ] (x1) -- (E1);
\foreach \from/\to in {x2/E2,x3/E3,xN/EN,xN1/EN1,xN2/EN2}
\draw[->, thick] (\from) -- (\to);
\foreach \from/\to in {x1/x2,x2/x3,xN2/xN1,xN/xN1,x1/xN}
\draw[<->, thick] (\from) -- (\to);
\end{tikzpicture}
}
\caption{Digraph of preferred transitions of the ratchet system for the rates (\ref{rate1}) on the outer and inner ring and where $a=1$.}
\label{digraphratchet}
\end{figure}
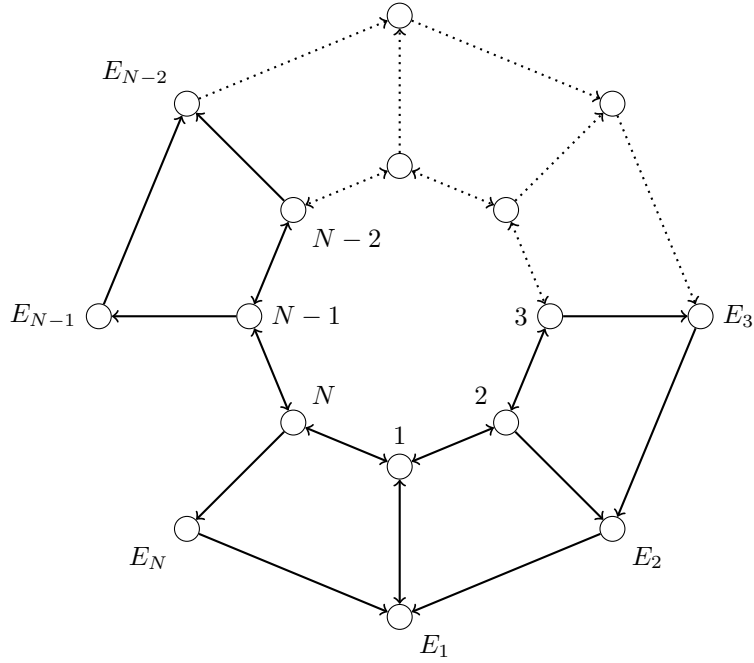
	
 We choose energies $E_i$ with order $E_1<\ldots < E_N$.  The transition rates on the outer ring ($n=0$) are
\be\label{rate1}
\lambda((i,0),(i+1,0))=e^{\frac{\beta}{2} (E_i-E_{i+1})}, \quad \lambda((i+1,0),(i,0))=e^{\frac{\beta}{2} (E_{i+1}-E_{i})}
\ee
The rates on the inner ring ($n=1$) are  constant,
\be\label{rate2}
\lambda((i,1),(i+1,1))=\lambda((i+1,1),(i,1))=1
\ee
The rings are connected with transition rates $\lambda((i,n),(i,1-n))=a$ for some $a>0$. (Separately, i.e., for $a=0$, the dynamics on both rings satisfy detailed balance, but with respect to another potential.) At low temperatures the transitions $(i,0)\rightarrow (i+1,0)$ become exponentially damped. \\ 

In the case when $a\ll1$, the rings get uncoupled and global detailed balance is restored for $a=0$. On the other hand, when the coupling $a\gg 1$ becomes very strong, the model is effectively running on a single ring with transition rates being the sum
\[
\lambda_{a=\infty}(i,i\pm 1) = e^{\frac{\beta}{2} (E_i-E_{i\pm 1})} + 1,\quad  \frac{\lambda_{\infty}(i,i+ 1)}{\lambda_{\infty}(i+ 1,i)} =  e^{\frac{\beta}{2} (E_i-E_{i+1})}
\] 
and thus again satisfying detailed balance but for inverse temperature $\beta/2$.\\ 
An interesting nonequilibrium situation arises when $a=1$, see Fig.~\ref{digraphratchet} for the graph of preferred transitions in this case.  We apply formula (\ref{eq:as2}) for the stationary distribution $\rho$. We have here $a(x,y)=1$ over all edges. Therefore in \eqref{aa}, the prefactor $A(x) = |M(x)|$  equals the number of in-trees for which $U(\mathcal{T}_x)$ is minimal. On the other hand, for the exponential factor in (\ref{eq:as2}), $\Gamma(x)-\Theta(x)$ becomes maximal and equal to zero for $x\in {\cal D} := \{(1,0),(i,1),i=1,\ldots,N\}$;  ($\Gamma(x)=0  =\Theta(x)$ for $x\in {\cal D}$).  As a consequence, at low temperatures, $\rho(x) \propto |M(x)|$ for $x\in {\cal D}$, and $\rho(y) \simeq |M(y)|\,e^{\beta \Gamma(y)}/{\cal Z}, \Gamma(y)<0$, is exponentially damped for $y\notin \cal D$. When the energy levels are equidistantly spaced in $[0,1]$, then $\varepsilon\propto 1/N$ so that $\beta\gg N$ is required for Theorem \ref{thm1} to apply.\\
To discover the most occupied state thus amounts to finding $x\in {\cal D}$ with the largest $|M(x)|$, the number of in-trees for a given state $x$ in the digraph. That number is given by the Matrix-Tree Theorem; see e.g. \cite{Tutte}.
We need the Laplacian matrix $L$ on the digraph $K^D$ and we erase the row and the  column corresponding to vertex $x$ to obtain the matrix $L_x$.  Then,
\be
A(x)= |M(x)| = \det L_x
\ee
The Laplacian of the digraph $K^D$ has a rather simple structure:
\be\nonumber
L = \bordermatrix{ 
 ~      &  (1,0)       & (2,0) & \hdots  &\hdots&(N,0)&(1,1)&  (2,1) &\hdots	& \hdots&(N,1) \\	\nonumber
(1,0)   &   1          &       &         &      &     &-1   &        &    	  &       &       \\	\nonumber
(2,0)   &  -1  &  1    &        &       &       &    &       &		  &    	  &	   \\   \nonumber
\vdots&      & \ddots& \ddots &       &       &    &       &		  & 	  &    \\   \nonumber
\vdots&      &       & -1     &1      &       &    &       &		  & 	  &    \\   \nonumber
(N,0)   &   -1  &       &        &    0  & 1     &    &       &		  & 	  &    \\\hline \nonumber
(1,1)   &   -1 &       &        &       &       & 3  & -1    &		  & 	  &   -1 \\	\nonumber
(2,1)   &      & -1    &        &       &       & -1 & 3     & \ddots   &	      &    \\   \nonumber 
\vdots&      &       & \ddots &       &       &    & \ddots& \ddots   & \ddots&    \\   \nonumber 
\vdots&      &       &        & \ddots&       &    &       & \ddots   & \ddots& -1 \\   \nonumber 
(N,1)   &      &       &        &       &-1     & -1   &       &          & -1    & 3\\}    \nonumber
\ee
The state for which the number of in-trees becomes maximal is $(1,0)$; it is easy to see that there are more combinations to form an in-tree to $(1,0)$ than to any other state $(i,1)$ on the inner ring.\\
The stationary ratchet current $J_R$ in the clockwise direction is the current over both rings, that is
\be
J_{R}=j((i+1,0),(i,0))+ j((i+1,1),(i,1))
\ee
with $j(x,y) = \lambda(x,y)\rho(x) -\lambda(y,x)\rho(y)$.  Of course the ratchet current also depends on  $N$ (ring size) and on the energy landscape.   Interestingly, the direction of the ratchet current is not decided by the second law, in contrast with the usual case for currents produced by thermodynamic forces.  For example, consider the following two trajectories to go around the ring: $\omega_1 = ((N,0),(N-1,0),\ldots,(1,0), (1,1),(N,1),(N,0))$ and $\omega_2 = ((N,0),(1,0),(1,1),(2,1),\ldots,(N,1),(N,0))$.  They wind in opposite direction while their entropy flux computed from \eqref{ldb} is exactly equal to $s(\omega_1)=s(\omega_2) = \beta(E_N - E_1) > 0$; Fig.~\ref{twotraj}. 

%\begin{figure}[t]
%\centering
%\hspace{-1cm}\includegraphics[,width=11 cm]{twotraj}
%
%\end{figure}
\begin{figure}[h]
\centerline{
\begin{tikzpicture}
\node [draw, shape= circle,label=90:$1$ ](x1) at (270:2)  {$$};
\node [draw, shape= circle,label=135:$2$] (x2) at (315 :2) {$$};
\node  [draw, shape= circle,label=180:$3$ ](x3) at (0:2) {$$};
\node [draw, shape= circle ,label=45:$N$ ] (xN) at (225:2) {$$};
\node  [draw, shape= circle,label=0:$N-1$ ](xN1) at (180:2) {$$};
\node  [draw, shape= circle,label=315:$N-2$ ](xN2) at (135:2) {$$};
\node  [draw, shape= circle ](xDOT1)  at (90:2) {};
\node [draw, shape= circle ] (xDOT2)  at (45:2) {};
\node [draw, shape= circle,label=320:$E_1$] (E1) at (270:4) {$$};
\node  [draw, shape= circle,label=320:$E_2$ ](E2) at (315 :4) {$$};
\node  [draw, shape= circle,label=0:$E_3$ ](E3) at (0:4) {$$};
\node [draw, shape= circle,label=225:$E_N$ ]  (EN) at (225:4) {$$};
\node [draw, shape= circle,label=180:$E_{N-1}$  ] (EN1) at (180:4) {$$};
\node [draw, shape= circle,label=135:$E_{N-2}$  ] (EN2) at (135:4) {$$};
\node  [draw, shape= circle ](DOT1)  at (90:4) {};
\node  [draw, shape= circle ](DOT2)  at (45:4) {};
\foreach \from/\to in {E2/E1, E3/E2,EN1/EN2}	
\draw [blue,-> , thick] (\from) -- (\to);
\draw[blue,dotted,-> , thick] (EN2)-- (DOT1);
\draw[blue,dotted,-> , thick] (DOT2)--(E3) ;
\draw[blue,dotted,->, thick ] (DOT1)--(DOT2)  ;
\draw[dotted,-, thick ] (xDOT1)--(DOT1)  ;
\draw[dotted,-, thick ] (xDOT2)--(DOT2)  ;
\draw[red,dotted, thick, -> ] (xDOT1) -- (xN2);
\draw[red,dotted, thick, -> ] (xDOT2) -- (xDOT1);
\draw[red,dotted, thick, -> ] (x3) -- (xDOT2);
\draw[red, thick, -> ] (EN) -- (E1);
\draw[blue, thick, -> ] (EN) -- (EN1);
\draw[thick,blue,-> ] (x1) -- (xN);
\foreach \from/\to in {x2/E2,x3/E3,xN1/EN1,xN2/EN2}
\draw[-, thick] (\from) -- (\to);
\foreach \from/\to in {x1/x2,x2/x3,xN2/xN1,xN1/xN}
\draw[->, thick,red] (\from) -- (\to);
  \path[blue,thick,->,every node/.style={font=\sffamily\small}]
(E1) edge [bend right] (x1);
  \path[red,thick,->,every node/.style={font=\sffamily\small}]
(E1) edge [bend left] (x1);
  \path[blue,thick,->,every node/.style={font=\sffamily\small}]
(xN) edge [bend right] (EN);
  \path[red,thick,->,every node/.style={font=\sffamily\small}]
(xN) edge [bend left] (EN);
\end{tikzpicture}
}
\caption{Trajectoreis $\omega_1$ and $\omega_2$ with the same entropy flux, yet in opposite directions.}
\label{twotraj}
\end{figure}
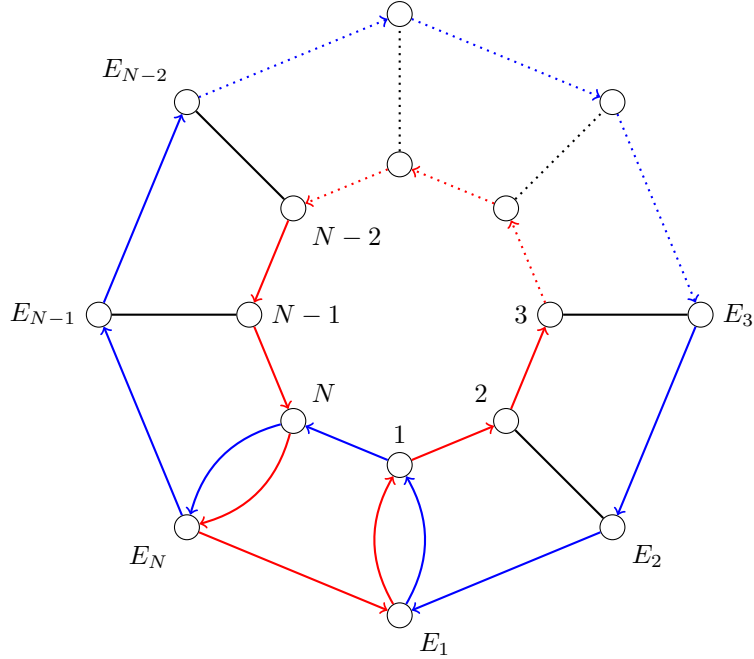

That makes the point for kinetic aspects, or here, the importance of $A(x)$ in \eqref{aa} for deciding the nature of the ratchet current.\\

Let us denote $x\simeq y$ when $x=y+O(e^{-\beta \epsilon})$. We take $x=(1,1)$ for which
 $\rho(1,1)\simeq\frac{1}{\mathcal{Z}}A((1,1))$.  Then,
\be\nonumber
j((2,1),(1,1)) \simeq \frac{1}{\mathcal{Z}}\left(A((2,1))-A((1,1))\right)
\ee 
Moreover,
\be\nonumber
j((2,0),(1,0)) \simeq \frac{A((2,0))}{\mathcal{Z}}
\ee
As  a consequence,
\be\nonumber
J_{R}\simeq\frac{1}{\mathcal{Z}}\left( \det L_{(2,1)}+\det L_{(2,0)} - \det L_{(1,1)}\right)
\ee
From the form of the Laplacian $L$, one finds that
\begin{enumerate}
\item $\det L_{(2,0)}=2\det B_{N-1} -3\det B_{N-2}-3$,
\item $\det L_{(1,1)}=\det B_{N-1}$,
\item $\det L_{(2,1)}=\det B_{N-2}+1$
\end{enumerate}
with 
\be
B_N =\left( \begin{array}{cccc}
   3  & -1    &        &              \\	
  -1  & 3    & \ddots &         \\
      & \ddots& \ddots & -1         \\
      &       & -1 & 3        \\ 
 \end{array}\right)
\ee
$B_N$ satisfies the recursion relation $\det B_N=3\det B_{N-1} - \det B_{N-2}$, where $\det B_2=8$ and $\det B_1=3$. Hence,
\be\nonumber\label{ratcur1}
J_{R}&\simeq&\frac{\det B_{N-1}-2\det B_{N-2}-2}{\mathcal{Z}}\\
&=&\frac{\det B_{N-2}-\det B_{N-3}-2}{\mathcal{Z}}
\ee
Since $\det B_2=8$ and $\det B_1=3$, we have $\det B_2>\det B_1+2$. Assume this holds for $B_N$, i.e., that $\det B_N>\det B_{N-1}+2$.  Then,
\be\nonumber 
\det B_{N+1}=3\det B_{N}-\det B_{N-1}>3\det B_{N}-\det B_{N-1}>2\det B_N-2>\det B_N-2
\ee
which proves that $J_{R}>0$, $\forall N\geq 4$ and $J_{R}=0$ when $N=3$. In other words the direction is clockwise.  The recursive relation allows to calculate $B_N$ for arbitrary $N$. The result for the ratchet current is shown in Fig. \ref{ratcur} (always up to order $O(e^{-\beta \epsilon})$). Notethat the current saturates as the system size $N$ increases. This is due to the fact that the stationary occupations also saturate, see Fig. \ref{ovn}. This is in turn a consequence of the fact that the ratio between the number of in-trees for a specific state versus the sum of in-trees over all (dominant) states saturates.
  
\newpage
\begin{figure}[h]
\centering
\hspace*{-1.3cm}\begin{subfigure}{.5\textwidth}
  \centering
  \includegraphics[width=\textwidth]{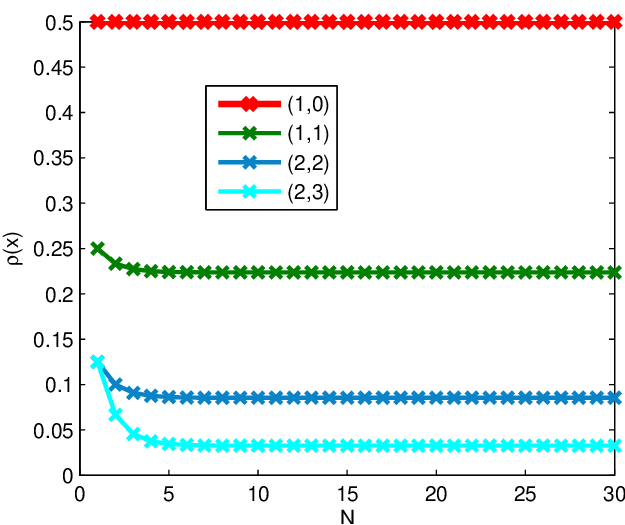}
  \caption{Some occupations as a function of the ring size $N$}
  \label{ovn}
\end{subfigure}%
\hspace*{1cm}\begin{subfigure}{.5\textwidth}
\includegraphics[width=\textwidth]{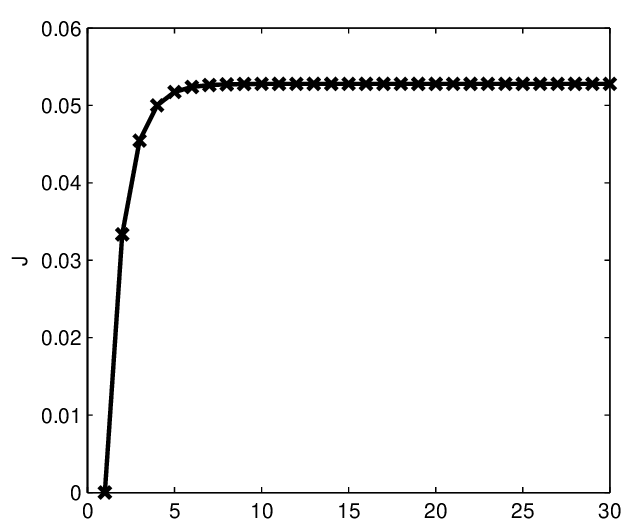}
\caption{The ratchet current $J_R$ as function of the size $N$ from \eqref{ratcur1}.}
  \label{ratcur}
\end{subfigure}
  \caption{}
\end{figure}

\section{Population inversion}\label{laser}
Lasers provide other examples of typical nonequilibrium phenomena. A laser medium consists of atoms  with the same energy spectrum. One can thus treat the system as an ensemble of particles that populate different energy levels. The ensemble is in contact with an equilibrium heat reservoir (the lattice) at inverse temperature $\beta$. The nonequilibrium condition arises from an external field that excites particles to a higher energy level. The resulting  population inversion is important for laser amplification.\\

Consider again a Markov jump process on  $\{1,2,\ldots,N\}$ with $N>2$ and define the transition rates as either purely relaxational or equalizing.  Between any neighboring levels $x \leftrightarrow y, |x-y|=1$, the relaxational dynamics follows detailed balance at inverse temperature $\beta$, and we take rates
\[
\lambda(x,y) = \psi(x,y)\,e^{-\frac{\beta}{2}(E(y)-E(x))},\quad  \{x,y\}\neq \{1,N\}
\]
for symmetric $\psi(x,y)=\psi(y,x)$ and energy levels $E(1)<E(2) <\ldots<E(N)$.
There are also transitions between the lowest and the highest level in the form of an equalizing process
\[
\lambda(1,N) =\lambda(N,1) = b  > 0
\]
with $b$ independent of temperature.

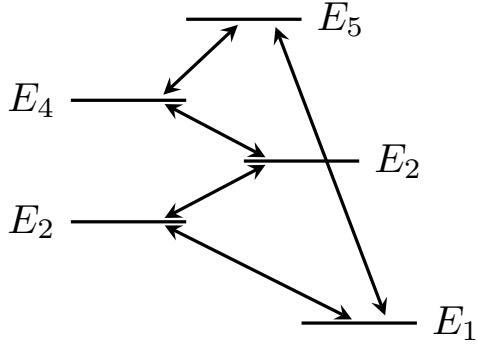
\begin{figure}[h]
\centerline{
  % Resize it to 5cm wide.
  \resizebox{7cm}{!}{
    \begin{tikzpicture}[
      scale=0.5,
      level/.style={thick},
      virtual/.style={thick,densely dashed},
      trans/.style={thick,<->,shorten >=2pt,shorten <=2pt,>=stealth},
      classical/.style={thin,double,<->,shorten >=4pt,shorten <=4pt,>=stealth}
    ]
    % Draw the energy levels.
	  \draw[level] (2cm,4em) -- (4cm,4em) node[right] {$E_5$};
	  \draw[level] (2cm,0em) -- (0cm,0em) node[left] {$E_4$};
	  \draw[level] (3cm,-3em) -- (5cm,-3em) node[right] {$E_2$};
    \draw[level] (2cm,-6em) -- (0cm,-6em) node[left] {$E_2$};
    \draw[level] (4cm,-11em) -- (6cm,-11em) node[right] {$E_1$};
    % Draw the transitions.
		\draw[trans] (1.5cm,0em) -- (3cm,4em);
    \draw[trans] (3.5cm,-3em) -- (1.5cm,0em);
    \draw[trans] (1.5cm,-6em) -- (3.5cm,-3em);
    \draw[trans] (5cm,-11em) -- (1.5cm,-6em);
    \draw[trans] (5.5cm,-11em)--(3.5cm,4em); 
    \end{tikzpicture}
  }
}
\caption{A 5-level diagram with the possible transitions.}
\end{figure}

The question is to specify the low temperature stationary distribution, and to understand the role of the  equalizer.  In particular we show how to organize the choice so that  a specific level $i\neq 1$ gets most populated.\\  
We take $\psi(i-1,i)=\psi(i,i-1)=a\,e^{-\beta F/2}$  for $F>0$ and for some $a=e^{o(\beta)}$. We also set $\psi(x,x+1)=e^{o(\beta)}$ for the other states $x\neq i-1$. The digraph is strongly connected so that $\Theta(x)=0$.  From \eqref{eq:as2} the stationary occupations thus become $\rho(x)\simeq \frac{1}{\mathcal{Z}} A(x)\, e^{\beta \Gamma(x)}$
 and the lifetimes appear to be  decisive. For $x=1$ we have $\Gamma(1) = 0$.  For $x=i$ it depends:  When
 $ F\geq E_{i+1} - E_{i-1}$, then $2\Gamma(i) = E_{i+1}-E_i >0$; when $E_i-E_{i-1} <  F \leq E_{i+1}-E_{i-1}$, then $2\Gamma(i) = F- E_{i}+E_{i-1} > 0$ and the level $i$ will be most occupied.  However, when $F=E_i-E_{i-1}$, then 
from formula (\ref{eq:as2}),
\be\nonumber
&&\rho(1)\simeq\frac{ \prod_{x\neq 1,i} \psi(x-1,x)}{\mathcal{Z}} \,a \\
&&\rho(i)\simeq\frac{ \prod_{x\neq 1,i} \psi(x-1,x)}{\mathcal{Z}} b \\\nonumber
&&\rho(x)\simeq 0, \quad x\neq \{1,i\}
\ee
and it depends on the ratio $a/b$ whether it is level $x=1$ or level $x=i$ that is to win the largest population.\\
 Note that $b=0$ would restore detailed balance and then no population inversion is possible. That shows again the important dependence on the reactivities for nonequilibrium statistics.

\section{Conclusion and outlook}

In contrast to thermal equilibrium, nonequilibrium steady states are determined also by kinetic aspects of the dynamics which are encoded in different reactivities of stochastic transitions between states of the system. Some more insight and remarkable simplifications can be obtained in the low temperature regime where the asymptotic behavior is related to dominant preferences within the Markovian transition rules. 

In this paper we have derived and discussed an improved asymptotic expression for the stationary occupations in nonequilibrium multilevel systems at low temperatures, explicitly including the $\beta-$subexponential corrections given by a restricted Kirchhoff tree formula. We have illustrated its usefulness by evaluating the ratchet current in a model of a
flashing potential for which subexponential corrections play a crucial role by removing the ``false'' degeneracies originating in a ``naive'' exponential asymptotic analysis. This is a typical example of a problem where a purely thermodynamic approach based on the Second Law inequality does not suffice to determine the direction of a current. 
Related is the phenomenon of population inversion controled by the variable transition reactivities (sometimes referred to as the Landauer blowtorch theorem) which we have demonstrated via our second example.\\

Although the geometric representations prove useful for the low temperature analysis of small stochastic systems, the complexity fast increases with the number of degrees of freedom. It still remains an open question whether this approach can be extended to a useful framework for spatially extended systems with local transition rules. That could open a new route, e.g., for studies on dynamical phase transitions in the low temperature domain.

\end{document}